\documentclass[11pt]{article}
\usepackage[a4paper, total={6in, 8in}]{geometry}
\usepackage{authblk}
\usepackage{graphicx} 
\usepackage{amsmath}
\usepackage{amsthm}
\usepackage{amssymb}
\usepackage{enumerate}
\usepackage{url}
\usepackage{placeins}
\usepackage{natbib}
\usepackage[colorlinks=true, allcolors=blue]{hyperref}
\setcitestyle{authoryear,aysep={}}
\newtheorem{prop}{Proposition}
\newtheorem{remark}{Remark}
\newcommand{\ind}{\perp\!\!\!\!\perp}
\allowdisplaybreaks

\title{Efficient Bayesian Inference for Benter Models\\ on Ranked Data}
\author{Michael Pearce}
\affil{Department of Mathematics and Statistics\\ Reed College}
\date{August 2026}

\begin{document}

\maketitle

\bigskip
\begin{abstract}
The Benter model for partial and complete rankings generalizes the well-known Plackett-Luce by attaching rank-level dampening parameters that permit some ranking stages to be noisier than others. This extension is empirically important for domains ranging from horse racing to ranked-choice elections. However, model fitting is made challenging by the fractional powers these dampening parameters introduce to the likelihood, breaking the conjugacy underlying existing Plackett-Luce samplers. This paper develops an efficient Bayesian estimation procedure for the Benter model. We introduce a two-part augmentation scheme using positive $\alpha$-stable and exponential auxiliary variables that linearizes the intractable normalizers in the Benter likelihood and yields closed-form Gibbs updates. A simulation study confirms efficient estimation, accurate parameter recovery, and nominal credible-interval coverage across sample sizes and item counts. We illustrate the algorithm on complete and partial rankings from survey preference and ranked-choice election datasets. 
\end{abstract}

\noindent%

\bigskip

\section{Introduction}

Ordinal preference data arises in a wide variety of applications, including in psychological and behavioral experiments, ranked-choice elections, and sports. In such cases, we often assume that independent processes (e.g., individuals, voters, athletic games) rank various items (e.g., options, candidates, teams). Parametric statistical models are useful to estimate or predict population-level preference parameters, potentially with their associated uncertainty, such as an overall ordering of the items from best to worst. An overview of such models can be found in reviews by \cite{Marden1996} and \cite{alvo2014statistical}.

The Plackett-Luce model \citep{plackett1975analysis,luce1959individual} is widely-used and studied as a distribution for rankings: Let $\mathcal{J} = \{1,\dots,J\}$ index a set of items and let $\pi$ be an observed ranking, where $\pi(r) \in \mathcal{J}$ denotes the item assigned rank $r$ in $\pi$. Furthermore, let $\lambda\in\mathbb{R}^J_{>0}$ be a vector of item-specific \textit{worth} parameters, where higher values indicate a more-preferred object. Then, the Plackett-Luce assumes $\pi$ is drawn according to
\begin{equation}\label{eq:PL}
    p[\pi|\lambda] = \prod_{r=1}^{J-1} \frac{\lambda_{\pi(r)}}{\sum_{k=r}^J\lambda_{\pi(k)}}
\end{equation}
where the $J^\text{th}$ product term is identically one and thus not included. Despite the relative simplicity of the model, estimation is challenging due to the denominators $\sum_{k=r}^J\lambda_{\pi(k)}$ that are burdensome to compute at scale and do not permit direct maximization of the likelihood. In the frequentist setting, a Minorization-Maximization (MM) algorithm was applied by \cite{Hunter2004} for more efficient maximum likelihood estimation. In the Bayesian setting, \cite{caron2012efficient} introduced a family of Exponential auxiliary variables with rates each equal to the denominator in (\ref{eq:PL}) for each ranking and rank-level. In tandem with independent Gamma priors on each worth parameter $\lambda_j$, the augmented joint likelihood undergoes massive simplification and enables closed-form Gibbs sampling. Their work was later extended to the partial rankings and mixture model settings in \cite{mollica2017bayesian} and implemented in the \texttt{PLMIX}\footnote{As of this writing, \texttt{PLMIX} has been removed from CRAN but is still available in archived form.} R package \citep{Mollica2020PLMIX}.

The Benter model \citep{benter1994computer} is an important extension of the Plackett-Luce. The Benter adds rank-level specific \textit{dampening} parameters, $\alpha_r\in(0,1]$, $1\leq r\leq J-1$, which introduce additional variation in the selections made at rank $r$ when $\alpha_r<1$; smaller values indicate greater variability. That is, the model permits cases where the top choice is made and subsequent choices are, potentially, noisier. Specifically, the Benter assumes $\pi$ is drawn according to
\begin{equation}\label{eq:Benter_simple}
    p[\pi|\lambda,\alpha] = \prod_{r=1}^{J-1} \frac{\lambda_{\pi(r)}^{\alpha_r}}{\sum_{k=r}^J\lambda_{\pi(k)}^{\alpha_r}}.
\end{equation}
Although originally developed for horse betting, the Benter model has proven useful in ranked-choice election modeling \citep{gormley2008exploring,gormley2008mixture,pearce2026can}, where voters may be assumed to have limited information on lower-ranked candidates and thus rank them more randomly than their most-preferred. 

Efficient estimation of the Benter model is elusive. Rank-specific dampening parameters further complicate the likelihood, making likelihood maximization more difficult and eliminating the conjugacy obtained by \cite{caron2012efficient}. \cite{gormley2008exploring} fit Benter models via an MM algorithm and produce approximate standard errors in an EM-based approach. \cite{pearce2026can} obtain point estimates for Benter models via an EM algorithm that employs numerical optimization in R. \cite{vanichbuncha2017modelling} compares an MM-based approach to maximum likelihood estimation with numerical optimization, finding the latter to be faster (but still slow). We are unaware of any Bayesian methods to fit Benter models, nor of any software to fit them in either the frequentist or Bayesian context. This has left the Benter model without full Bayesian uncertainty quantification, forcing practitioners to rely on point estimates and asymptotic standard errors from EM-based approaches.

In this paper, we propose a computationally-efficient Bayesian estimation algorithm for the Benter model. The algorithm is based on a two-part data augmentation scheme that yields a blocked Gibbs sampler for the model parameters. We demonstrate that the algorithm quickly and accurately recovers a full posterior distribution of the model parameters in a simulation study, permitting both point estimation and inference based on the Benter model. Furthermore, we exhibit its use on two real datasets in preferences and ranked-choice elections, including on complete and top-$r$ rankings.

The paper proceeds as follows. Section \ref{section:BayesianBenter} introduces the Bayesian Benter model, including the data augmentation scheme and prior specification. We develop an efficient Gibbs sampler in Section \ref{section:Sampler}, and demonstrate its speed and accuracy in analyses of simulated and real data in Section \ref{section:DataAnalyses}. The paper concludes with a brief discussion in Section \ref{section:Discussion}.

\section{The Bayesian Benter model}\label{section:BayesianBenter}

We first re-introduce the Benter model in greater generality of notation. Let $\mathcal{J} = \{1,\dots,J\}$ index a set of items and $\pi_1,\dots,\pi_n$ be observed rankings, where $\pi_i(r) \in \mathcal{J}$ denotes the item in rank $r$ of voter $i$.
Let $r_i=|\pi_i|$ be the length of ranking $\pi_i$. When $r_i=J$ the ranking is called \textit{complete}, and when $r_i<J-1$ the ranking is a top-$r_i$ \textit{partial} ranking.\footnote{Note that when $r_i=J-1$, the unranked item is naturally in last place, making such rankings complete.} 
Furthermore, we define a \textit{risk set} $R_{ir}$ for voter $i$ in stage $r$ according to
$$
R_{ir} = \mathcal{J} \setminus \{\pi_i(1),\dots,\pi_i(r-1)\}
$$
which denotes the set of items not yet ranked in advance of stage $r$, and define $w_{ir}\equiv \pi_i(r)$ as the stage winner.

The Benter model \citep{benter1994computer} assigns each item $j\in\mathcal{J}$ a positive \textit{worth} parameter $\lambda_j$, and each ranking stage a \textit{dampening} parameter $\alpha_r\in(0,1]$. The likelihood of observing ranking $\pi_i$ is assigned:
\begin{equation}\label{benter_likelihood}\
    p(\pi_i \mid \lambda, \alpha) \;=\; \prod_{r=1}^{\min(r_i,J-1)} \frac{\lambda_{w_{ir}}^{\alpha_r}}{\sum_{j \in R_{ir}} \lambda_j^{\alpha_r}}.
\end{equation}
Note that the product terminates at $r = \min(r_i,J-1)$ because the $J^\text{th}$ factor is identically one; in what follows, for notational simplicity we will always terminate the product at $r_i$. Setting all $\alpha = 1$ recovers the Plackett–Luce model; as $\alpha_r\to0$ stage $r$ approaches uniform choice over $R_{ir}$.

Without further restrictions, the Benter likelihood in (\ref{benter_likelihood}) is unidentifiable. First, multiplicative rescaling $\lambda_j \mapsto c\lambda_j$ cancels within each stage. Second, the power transformation $(\lambda_j, \alpha_r) \mapsto (\lambda_j^{\beta}, \alpha_r/\beta)$ leaves every factor unchanged. We address these issues by constraining $\sum_j\lambda_j=1$ and fixing $\alpha_1\equiv 1$. We note that the former constraint on $\lambda$ may be imposed indirectly in the Bayesian context by post-hoc normalization of posterior samples of $\lambda$.

In summary, the Benter model has $2J-3$ free parameters: $J$ \textit{worth} parameters (of which only $J-1$ are free) and $J-2$ \textit{dampening} parameters (since $\alpha_1\equiv1$ and stage $J$ contributes no likelihood factor).

\subsection{Prior specification}

Each object worth, $\lambda_j$, is a positive real. Thus, we assign independent Gamma priors,
$$\lambda_j\overset{iid}\sim \text{Gamma}(a,b), \qquad j = 1,\dots,J,
$$
parameterized by shape $a$ and rate $b$. The Gamma priors mimic those imposed by \cite{caron2012efficient} and \cite{mollica2017bayesian} for the Plackett-Luce distribution. We will soon see that these priors permit independent Gibbs sampling of each $\lambda_j$ via their full conditionals. 

Rank-level dampening parameters, $\alpha_r$, are restricted to the unit interval. Thus, we assign independent Beta priors,
$$\alpha_r\overset{iid}\sim \text{Beta}(c,d), \qquad r = 2,\dots,J-1,$$
where $c=d=1$ imposes a Uniform prior on the unit interval. To our knowledge, no continuous prior yields a recognizable conditional for $\alpha_r$. Thus, we choose a simple prior on the proper domain and update via Metropolis-Hastings.

\subsection{Augmented Benter model}

Posterior computation under the Benter likelihood in  (\ref{benter_likelihood}) is impeded by the normalizing sums $\sum_{j\in R_{ir}} \lambda_j^{\alpha_r}$, which are neither conjugate to the Gamma prior nor separable across items. We will sidestep this challenge by introducing two auxiliary variables. 

As a preliminary, for $\alpha\in(0,1]$ let $f_\alpha$ denote the density of the positive (one-sided) $\alpha$-stable distribution, $\text{PS}(\alpha)$, on $(0,\infty)$ \citep{devroye2009random}. No explicit formula characterizes $f_\alpha$ for general $\alpha$. Instead, the PS random variable may be characterized by its Laplace transform,
\begin{equation}\label{PS_laplace}
  \int_0^\infty e^{-ts} f_\alpha(s)\, ds = e^{-t^{\alpha}}, \qquad t \ge 0 .  
\end{equation}
Furthermore, differentiating (\ref{PS_laplace}) in $t$ yields the companion identity
\begin{equation}\label{PS_laplace_differentiated}
    \int_0^\infty s\, e^{-ts} f_\alpha(s)\, ds = \alpha\, t^{\alpha - 1} e^{-t^{\alpha}}.
\end{equation}
Note that when $\alpha=1$, the PS random variable is degenerate at $1$.

We are now ready to present an augmented Benter model with two families of latent variables, $S$ and $V$. Assume the following generative model for rankings $\pi$ given parameters $\lambda,\ \alpha$: For each voter $i$ and rank level $r$,
\begin{equation}\label{augmented_benter}
\begin{aligned}
S_{irj} \mid \alpha_r &\overset{ind}\sim \mathrm{PS}(\alpha_r), \qquad \qquad\qquad\qquad  j \in R_{ir}, \\[2pt]
V_{ir} \mid S, \lambda &\overset{ind}\sim \mathrm{Exp}\!\Big( \textstyle\sum_{k \in R_{ir}} \lambda_k S_{irk} \Big)\\
\pi_{ir}|S,\lambda &\sim \text{Categorical over $R_{ir}$ with Pr}(w_{ir}=j)\\
\end{aligned}
\end{equation}
The augmented joint likelihood of $\pi,V,S|\lambda,\alpha$ is thus
\begin{align}
    p(\pi,V,S|\lambda,\alpha) &= p(\pi,V|S,\lambda)p(S|\alpha)\nonumber \\
    &=\prod_{i=1}^n\prod_{r=1}^{r_i}\Big[\frac{\lambda_{w_{ir}}S_{ir{w_{ir}}}}{\sum_{k \in R_{ir}} \lambda_k S_{irk}}\Big(\sum_{k \in R_{ir}} \lambda_k S_{irk}\Big)e^{-V_{ir}\sum_{k \in R_{ir}} \lambda_k S_{irk}}\prod_{j\in R_{ir}}f_{\alpha_r}(S_{irj})\Big]\nonumber \\
    &=\prod_{i=1}^n\prod_{r=1}^{r_i}\Big[\lambda_{w_{ir}}S_{ir{w_{ir}}}\prod_{k\in R_{ir}}e^{-V_{ir}\lambda_k S_{irk}}f_{\alpha_r}(S_{irk})\Big].\label{augmented_joint_simplified}
\end{align}
We demonstrate the augmented Benter model in (\ref{augmented_benter}) is valid via Proposition \ref{prop_validity_aug_benter}.
\begin{prop}[Validity of Augmented Benter] \label{prop_validity_aug_benter}
    Let $\alpha_r\in(0,1]$ and $\lambda_j>0$, $j=1,\dots,J$ be fixed. Then, marginalizing out $(V,S)$ in the augmented Benter model in equation (\ref{augmented_benter}) recovers the Benter likelihood in (\ref{benter_likelihood}).
\end{prop}
\begin{proof}
    Note that
    \begin{align*}
      p(\pi|\lambda,\alpha) &= \int_V \int_S p(\pi,S,V|\lambda,\alpha)\ dS\ dV\\
      &= \int_V \int_S\prod_{i=1}^n\prod_{r=1}^{r_i}\Big[\lambda_{w_{ir}}S_{ir{w_{ir}}}\prod_{k\in R_{ir}}e^{-V_{ir}\lambda_k S_{irk}}f_{\alpha_r}(S_{irk})\Big]\ dS\ dV \\
      &= \prod_{i=1}^n\prod_{r=1}^{r_i}\int_{V_{ir}=0}^\infty \lambda_{w_{ir}} \int_{S_{irw_{ir}}=0}^\infty S_{ir{w_{ir}}}e^{-V_{ir}\lambda_{w_{ir}} S_{ir{w_{ir}}}}f_{\alpha_r}(S_{ir{w_{ir}}}) \ dS_{irw_{ir}} \times \\
      &\qquad \qquad\prod_{k\in R_{ir}\setminus\{w_{ir}\}}\int_{S_{irk}=0}^\infty e^{-V_{ir}\lambda_k S_{irk}}f_{\alpha_r}(S_{irk})\ dS_{irk} \ dV_{ir},\\
      \intertext{where the second line holds via (\ref{augmented_joint_simplified}). For each integral on $S_{irk}$, $k\in R_{ir}\setminus\{w_{ir}\}$, we utilize the identity in (\ref{PS_laplace}) with $t=\lambda_kV_{ir}$; for the integral on $S_{irw_{ir}}$ we utilize (\ref{PS_laplace_differentiated}) with $t=\lambda_{w_{ir}}V_{ir}$:}
      &= \prod_{i=1}^n\prod_{r=1}^{r_i}\int_{V_{ir}=0}^\infty \lambda_{w_{ir}}\alpha_r(\lambda_{w_{ir}}V_{ir})^{\alpha_r-1}e^{-(\lambda_{w_{ir}}V_{ir})^{\alpha_r}}\prod_{k\in R_{ir}\setminus\{w_{ir}\}}e^{-(\lambda_kV_{ir})^{\alpha_r}} \ dV_{ir}\\
      &= \prod_{i=1}^n\prod_{r=1}^{r_i}
      \lambda_{w_{ir}}^{\alpha_r} \int_{V_{ir}=0}^\infty \alpha_r V_{ir}^{\alpha_r-1}e^{-V_{ir}^{\alpha_r}\sum_{k\in R_{ir}}\lambda_k^{\alpha_r}} \ dV_{ir}\\
      \intertext{We evaluate this integral via $u$-substitution with $u=V_{ir}^{\alpha_r}$ and $du=\alpha_rV_{ir}^{\alpha_r-1}dV_{ir}$:}
      &= \prod_{i=1}^n\prod_{r=1}^{r_i}
      \lambda_{w_{ir}}^{\alpha_r} \int_{u=0}^\infty e^{-u\sum_{k\in R_{ir}}\lambda_k^{\alpha_r}} \ du \\
      &= \prod_{i=1}^n\prod_{r=1}^{r_i}\frac{\lambda_{w_{ir}}^{\alpha_r}}{\sum_{k\in R_{ir}}\lambda_k^{\alpha_r}},
    \end{align*}
    which is identical to the Benter likelihood in (\ref{benter_likelihood}).
\end{proof}

\section{An efficient Metropolis-within-Gibbs sampler}\label{section:Sampler}

In this section, we will develop an efficient Bayesian estimation procedure for the Benter model in \ref{benter_likelihood}. We start with presenting the augmented joint distribution of data $\pi$, augmented variables $S$ and $V$, and parameters $\lambda$ and $\alpha$, followed by the associated full conditional distributions. Then, we present a blocked Gibbs sampler.

\subsection{Augmented joint distribution and full conditionals}
    
The augmented joint distribution of $(\lambda,\alpha,V,S,\pi)$ may be written:
\begin{align}
    p(\lambda,\alpha,V,S,\pi) &= p(V,S,\pi|\lambda,\alpha)p(\lambda)p(\alpha)\nonumber\\
    &= \prod_{i=1}^n\prod_{r=1}^{r_i}\Big[\lambda_{w_{ir}}S_{ir{w_{ir}}}e^{-V_{ir}\sum_{k \in R_{ir}}\lambda_k S_{irk}}\prod_{k\in R_{ir}}f_{\alpha_r}(S_{irk})\Big] \times \label{augmented_joint} \\
    &\qquad \prod_{j=1}^J \frac{b^a}{\Gamma(a)}\lambda_j^{a-1}e^{-b\lambda_j}\times\prod_{r=2}^{J-1}\frac{\Gamma(c+d)}{\Gamma(c)\Gamma(d)}\alpha_r^{c-1}(1-\alpha_r)^{d-1}.\nonumber
\end{align}
Let $n_j=\sum_{i=1}^n\sum_{r=1}^{r_i}I\{w_{ir}=j\}$. We then derive the full conditionals from equation (\ref{augmented_joint}):
\begin{align}
p(\lambda_j|\lambda_{-j},\alpha,V,S,\pi) &\propto \lambda_j^{a-1}e^{-b\lambda_j}\times \prod_{i=1}^n\prod_{r=1}^{r_i}\lambda_{w_{ir}}e^{-\sum_{k\in R_{ir}}\lambda_k S_{irk}V_{ir}}\nonumber\\
    &= \lambda_j^{a-1}e^{-b\lambda_j}\times \lambda_j^{n_j}e^{-\lambda_j\sum_{(i,r): j\in R_{ir}} S_{irj}V_{ir}}\nonumber\\
    &\propto \text{Gamma}\Big(a+n_j, \ b+\sum_{(i,r): j\in R_{ir}}S_{irj}V_{ir}\Big),\label{fullcond_lambda}\\
p(V_{ir}|\lambda,\alpha,V_{-ir},S,\pi) &\propto e^{-V_{ir}\sum_{j\in R_{ir}}\lambda_jS_{irj}}\nonumber\\
    &\propto \text{Exp}\Big(\sum_{j\in R_{ir}}\lambda_jS_{irj}\Big),\label{fullcond_V}\\
p(S_{irj}|\lambda,\alpha,V,S_{-irj},\pi) &\propto  \begin{cases}
    e^{-V_{ir}\lambda_j S_{irj}}f_{\alpha_r}(S_{irj}), &j\neq w_{ir}\\
    S_{irj}e^{-V_{ir}\lambda_j S_{irj}}f_{\alpha_r}(S_{irj}), &j=w_{ir}
    \end{cases},\label{fullcond_S}\\
p(\alpha_r|\lambda,\alpha_{-r},V,S,\pi) &\propto \alpha_r^{c-1}(1-\alpha_r)^{d-1}\prod_{i=1}^n\prod_{k\in R_{ir}} f_{\alpha_r}(S_{irk}).\label{fullcond_alpha}
\end{align}
As can be seen, the full conditionals for $\lambda$ and $V$ yield recognizable distributions. The full conditional on $S$ is an \textit{exponentially tilted stable (ETS)} distribution when $j\neq w_{ir}$ and a \textit{size-biased ETS} distribution when $j=w_{ir}$. Neither of these distributions has a closed-form density. However, both admit exact sampling as will be shown in Section \ref{sampling_s_details}.

The full conditional for $\alpha$ exposes two issues: First, it requires pointwise evaluation of the stable density, which is not immediately available. Second, the $n|R_{ir}|$ conditionally iid PS observations imply that $p(\alpha_r|\lambda,V,S,\pi)$ concentrates at a scale of order $(n|R_{ir}|)^{-1/2}$, which is much tighter than the marginal posterior of $\alpha_r$ given the observed data alone. Conditioning on auxiliaries $S$ that were themselves generated from the current $\alpha_r$ pins the exponent in place, and the chain traverses the marginal posterior only by slow diffusion.

Therefore, we update each $\alpha_r$ jointly with the stage-$r$ auxiliaries, following the collapsed and partially collapsed Gibbs frameworks of \cite{liu1994collapsed} and \cite{van2008partially}. The block is $(\alpha_r, V_{\cdot r}, S_{\cdot r \cdot})$ and conditions only on quantities outside it: $\lambda$ and $\pi$. This joint conditional factorizes according to
\begin{equation}\label{block_conditionals}
    p(S_{\cdot r \cdot}, V_{\cdot r}, \alpha_r \mid \lambda, \pi)
= p(S_{\cdot r \cdot} \mid V_{\cdot r}, \alpha_r, \lambda, \pi)\;
p(V_{\cdot r} \mid \alpha_r, \lambda, \pi)\;
 p(\alpha_r \mid \lambda, \pi)
\end{equation}
and will be sampled by composition. Crucially, the stage-$r$ auxiliaries are members of the block and are therefore marginalized, not conditioned upon, when $\alpha_r$ is drawn. 

To complete the block updates, we must derive the three factors in (\ref{block_conditionals}). The first term on $S_{irj}$ is precisely the full conditional in (\ref{fullcond_S}), for which we discuss sampling in Section \ref{sampling_s_details}. For the second term on $V_{\cdot r}$, we follow similar logic to the proof of Proposition \ref{prop_validity_aug_benter}:
\begin{align*}
    p(V_{ir} \mid \alpha_r, \lambda, \pi) &= \int_S p(V_{ir},S_{ir\cdot}|\alpha_r,\lambda,\pi) dS\\
    &\propto \int_S S_{irw_{ir}}\prod_{k\in R_{ir}}e^{-V_{ir}\lambda_kS_{irk}}f_{\alpha_r}(S_{irk}) dS\\
    &\propto V_{ir}^{\alpha_r-1}e^{-V_{ir}^{\alpha_r}\sum_{j\in R_{ir}}\lambda_j^{\alpha_r}}\\
    &\propto \text{Weibull}\Big(\alpha_r, (\sum_{j\in R_{ir}}\lambda_j^{\alpha_r})^{-1/\alpha_r}\Big).
\end{align*}
Third and last, the term on $\alpha_r$ is
\begin{align}
    p(\alpha_r \mid \lambda, \pi)&\propto p(\alpha_r)p(\pi|\lambda,\alpha_r)\nonumber\\
    &\propto \alpha_r^{\,c-1}(1 - \alpha_r)^{\,d-1} \prod_{i=1}^{n} \frac{\lambda_{w_{ir}}^{\,\alpha_r}}{\sum_{j \in R_{ir}} \lambda_j^{\,\alpha_r}}, \qquad \alpha_r \in (0,1),\label{alpha_MH}
\end{align}
which may be sampled from approximately via Metropolis-Hastings.

\subsection{Exact sampling of $S_{irj}$ from its full conditional}\label{sampling_s_details}

Recall that the full conditional of $S_{irj}$ is an ETS distribution when $j\neq w_{ir}$ and a size-biased ETS when $j=w_{ir}$. In the former case, $S_{irj}$ may be sampled by the double-rejection algorithm of \cite{devroye2009random}, which draws exactly from an ETS in $O(1)$ expected time. We use the R function \texttt{retstable} in the \texttt{copula} package \citep{copula_R}. In the latter case required a size-biased ETS, we use the following decomposition:
\begin{prop}
Suppose $S$ has density proportional to $s f_\alpha(s) e^{-hs}$ where $\alpha \in (0,1)$ and $h > 0$. If $X \sim \mathrm{ETS}(\alpha, h)$, with density proportional to $f_\alpha(s) e^{-hs}$ and $Y \sim \mathrm{Gamma}(\text{shape}=1 - \alpha, \text{rate}=h)$ such that $X\ind Y$, then $S\overset{d}=X+Y$.
\end{prop}
\begin{proof}
By (\ref{PS_laplace_differentiated}), $S$ has normalizing constant $\frac{1}{\alpha h^{\alpha-1}e^{-h^\alpha}}$. Thus, the Laplace transform of $S$ is
$$\mathbb{E}[e^{-tS}]=\frac{\int_0^\infty e^{-ts}sf_\alpha(s)e^{-hs}ds}{\alpha h^{\alpha-1}e^{-h^\alpha}}\overset{(\ref{PS_laplace_differentiated})}=\frac{\alpha(h+t)^{\alpha-1}e^{-(h+t)^\alpha}}{\alpha h^{\alpha-1}e^{-h^\alpha}}=\Big(\frac{h+t}{h}\Big)^{\alpha-1}e^{h^\alpha-(h+t)^\alpha}.$$
Next, note that the Laplace transform of $Y$ is
$$\mathbb{E}[e^{-tY}]=\Big(\frac{h}{h+t}\Big)^{1-\alpha}=\Big(\frac{h+t}{h}\Big)^{\alpha-1}.$$
Furthermore, by (\ref{PS_laplace}) the normalizing constant of $X$ is $\frac{1}{e^{-h^\alpha}}$ with Laplace transform
$$\mathbb{E}[e^{-tX}]=\frac{\int_0^\infty e^{-tx}e^{-hx}f_\alpha(x)dx}{e^{-h^\alpha}}\overset{(\ref{PS_laplace})}=\frac{e^{-(h+t)^\alpha}}{e^{-h^\alpha}}=e^{h^\alpha-(h+t)^\alpha}.$$
Thus,
$$\mathbb{E}[e^{-tS}]=\mathbb{E}[e^{-tY}]\times \mathbb{E}[e^{-tX}].$$
Since the product of Laplace transforms corresponds to the sum of independent variables and Laplace transforms on positive random variables are unique, $S\overset{d}=X+Y$, as desired.
\end{proof}

The full conditional of $S_{irw_{ir}}$ is thus drawn by summing independent draws of an ETS distribution and a Gamma distribution.

\subsection{Estimation Algorithm}

In short, we propose estimating the Benter model as follows: Given complete or partial rankings $\pi$ and hyperparameters $(a,b,c,d)$, initialize $\lambda_j=1$ for $j=1,\dots,J$, $\alpha_r=1$ for $r=2,\dots,J-1$, and auxiliary variables $V,S$. Then for each Gibbs sweep,
\begin{enumerate}
    \item \underline{Update Worths:} Draw $\lambda_j \overset{ind}\sim \mathrm{Gamma}\big(a + n_j,\; b + \sum_{(i,r):\, j \in R_{ir}} S_{irj} V_{ir}\big)$, $j=1,\dots,J$.
    \item \underline{Stage 1 Updates:} Set $S_{i1j} \equiv 1$ and draw $V_{i1} \sim \mathrm{Exp}(\sum_j \lambda_j)$ for each $i$.
    \item \underline{Remaining Stage Updates:} For $r=2,\dots,J-1$, block-update $(\alpha_r,V_{\cdot r},S_{\cdot r\cdot})$:
    \begin{enumerate}
        \item Update $\alpha_r\propto \alpha_r^{\,c-1}(1 - \alpha_r)^{\,d-1} \prod_{i=1}^{n} \frac{\lambda_{w_{ir}}^{\,\alpha_r}}{\sum_{j \in R_{ir}} \lambda_j^{\,\alpha_r}}$ via Metropolis-Hastings.
        \item Draw $V_{ir}\overset{ind}\sim \text{Weibull}\Big(\alpha_r, (\sum_{j\in R_{ir}}\lambda_j^{\alpha_r})^{-1/\alpha_r}\Big)$.
        \item Draw $S_{irj}\overset{ind}\sim \text{ETS}(\alpha_r,V_{ir}\lambda_j)$ for $j\in  R_{ir}\setminus\{w_{ir}\}$ and $S_{irw_{ir}}\overset{ind}\sim  \text{ETS}(\alpha_r,V_{ir}\lambda_{w_{ir}})+\text{Gamma}(1-\alpha_r,V_{ir}\lambda_{w_{ir}})$.
    \end{enumerate}
\end{enumerate}
Lastly, store the \textit{normalized} worths, $\frac{\lambda_j}{\sum_{j'}\lambda_{j'}}$ for $j=1,\dots,J$, and dampening parameters $\alpha_r$ for $r=2,\dots,J-1$ for each iteration after burn-in.

\begin{remark}
    Stage 1 parameters are updated separately from those in stages $r=2,\dots,J-1$. This is principally because $\alpha_r\equiv1$ and thus does not require (block) updating. We may thus update the stage 1 auxiliaries via their full conditionals in (\ref{fullcond_V}) and (\ref{fullcond_S}). For $S_{\cdot 1\cdot}$, recall that $f_{\alpha_1=1}$ is degenerate at 1 and thus $S_{\cdot 1\cdot}\equiv 1$. Last, the full conditional on $V_{i1}$ is $\text{Exp}(\sum_{j}\lambda_jS_{i1j})$ which reduces to $\text{Exp}(\sum_j\lambda_j)$ in Step 2.
\end{remark}
\begin{remark}
    To improve efficiency, in practice we perform Step 3(a) 5 times for every Gibbs sweep. To further improve efficiency, we utilize an adaptive Metropolis-Hastings process by changing $\tau$ in the symmetric $\text{Normal}(0,\tau^2)$ proposal distribution during burn-in. The tuning parameter $\tau$ is adjusted separately for each $r=2,\dots,J-1$, as posterior distributions for each $\alpha_r$ have different variances.
\end{remark}

\section{Data Analyses}\label{section:DataAnalyses}

We now demonstrate the algorithm's accuracy, efficiency, and usefulness via analyses of data, both simulated and real. Code to replicate all analyses can be found \href{https://github.com/pearce790/FastBayesianBenter}{here}.

\subsection{Simulation Study}

We demonstrate the accuracy and computational efficiency of the proposed estimation algorithm in a brief simulation study. In the study, we vary the number of items ranked $J=\{4,8,12\}$ and the number of voters $n\in\{200,400,800,1600\}$. Furthermore, we vary $\lambda$ according for
$$\lambda = (1,s,s^2,\dots,s^{J-1})$$
for $s\in\{2,4,6\}$, which is then normalized such that $\sum_{j=1}^J\lambda_j=1$. When $s$ is small, items are less well-separated in worth and thus there is more variation in the observed rankings around the ranking $\pi_0=\{J\prec J-1\prec\dots\prec1\}$; when $s$ is large the observed rankings will exhibit less variation around $\pi_0$. For each combination of $n$, $J$, and $s$, we draw data $10$ times. For each repetition, we draw dampening parameters according $\alpha_r\overset{iid}\sim\text{Unif}(0.5,1)$ and subsequently draw data $\pi_i\overset{ind}\sim\text{Benter}(\lambda,\alpha)$. The model is then estimated with hyperparameters $a=c=d=1$ and $b=0$ to induce flat priors on $\lambda$ (improper) and $\alpha$ (proper) and with $10{,}000$ Gibbs sweeps, of which the first 20\% are removed as burn-in. We record the computational speed and the mean absolute error (MAE) and empirical coverage of 95\% credible intervals (CI) for parameters $\lambda$ and $\alpha$. All simulations were run on a MacBook Pro with an Apple M2 Max chip and 32 GB RAM in R version 4.5.1.

Figure \ref{fig:speed} displays computational speed. We observe runtimes that increase roughly linearly in $n$ and polynomially in $J$. However, the baseline is small enough such that running $10{,}000$ iterations with modestly large numbers of items and voters requires only a few minutes to run. In fact, across all simulation settings we observe runtimes that make the algorithm feasible. 
\begin{figure}[t!]
    \centering
    \includegraphics[width=.8\linewidth]{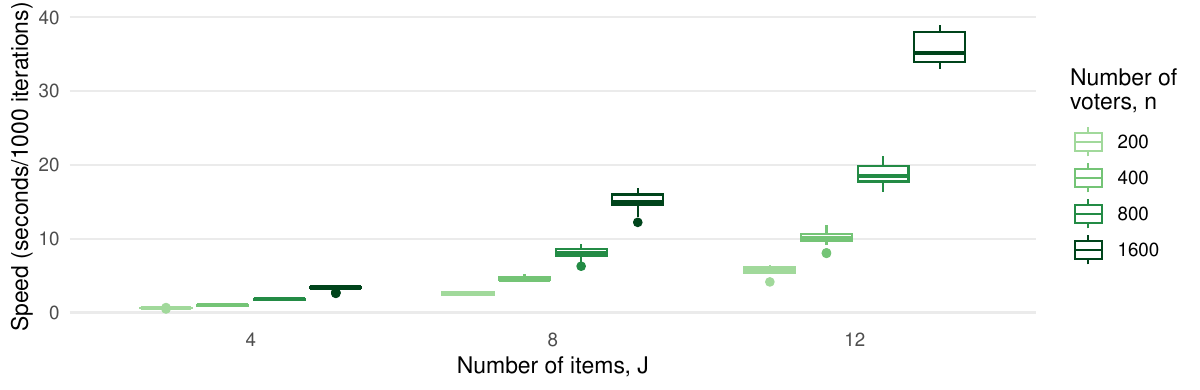}
    \caption{Computational speed of proposed Benter estimation algorithm across $J$, $n$, $s$. Speed is shown on y-axis by seconds per $1{,}000$ iterations. Simulations were run on a MacBook Pro with an Apple M2 Max chip and 32 GB RAM in R version 4.5.1.}
    \label{fig:speed}
\end{figure}

Figure \ref{fig:MAE} displays mean absolute error (MAE) in $\lambda$ (top) and $\alpha$ (bottom).
\begin{figure}[b!]
    \centering
    \includegraphics[width=.8\linewidth]{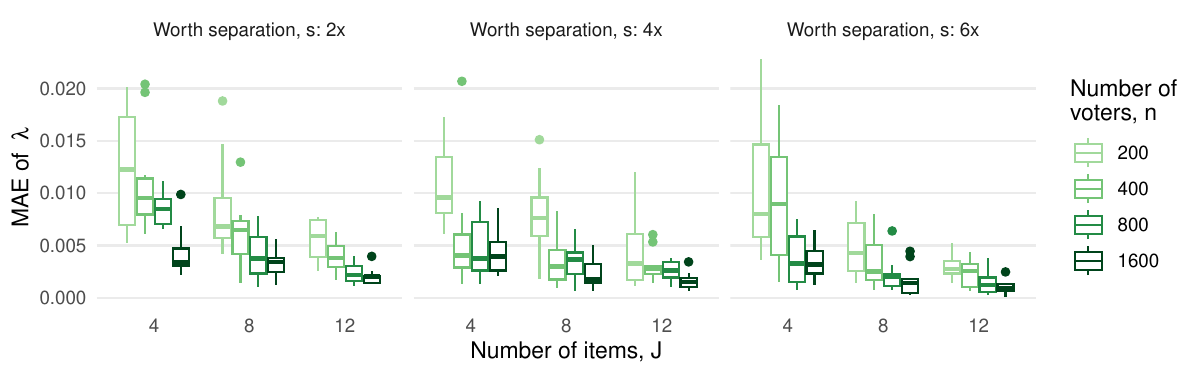}
    \includegraphics[width=.8\linewidth]{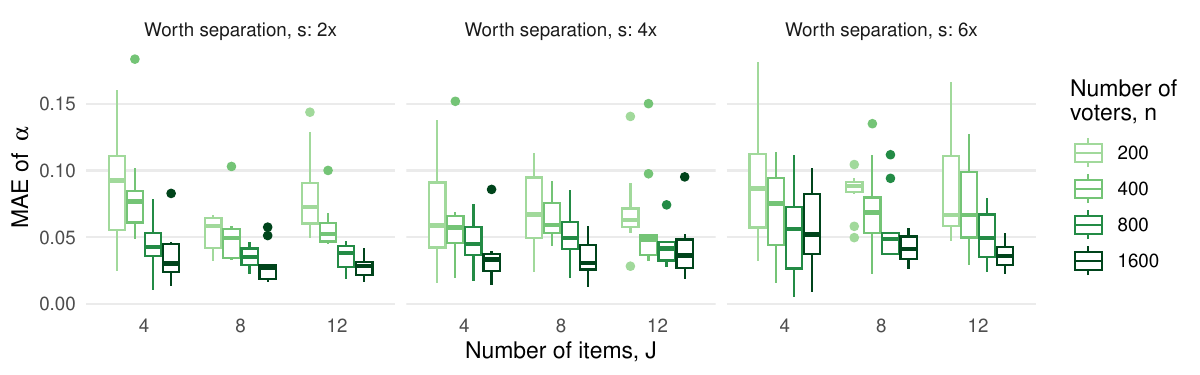}
    \caption{Mean Absolute Error in $\lambda$ (top) and $\alpha$ (bottom) across $J$, $n$, and $s$.}
    \label{fig:MAE}
\end{figure}
For $\lambda$, we notice that error is very small in magnitude regardless of $J$, $n$, and $s$. Error decreases as $n$ increases (consistency) and as $J$ increases (which is likely an artifact of the normalization of $\lambda$). There seems to be little effect of $s$ on error in $\lambda$. For $\alpha$, we again observe decreasing error in $n$. Unlike $\lambda$, there is no clear decrease in error for $\alpha$ as $J$ increases. We believe this may be explained by noting that dampening parameters are not normalized, and increasing $J$ by 1 corresponds to one additional $\alpha_r$ parameter. Finally, we observe that error generally increases in $s$. This may be expected, as more homogeneous rankings provides less information on the level of additional rank level-specific variation in preferences (which $\alpha$ specifically captures).

Last, we investigate accuracy of posterior distributions by examining the empirical coverage of 95\% credible intervals for $\lambda$ and $\alpha$ across sample size, $n$ (Figure \ref{fig:coverage}). We observe coverage close to the nominal 0.95 level, with some variation across settings which may be explained simply by Monte Carlo error.

Overall, the algorithm appears to be accurately and efficiently estimating model parameters in a variety of data regimes.
\begin{figure}[h!]
    \centering
    \includegraphics[width=.8\linewidth]{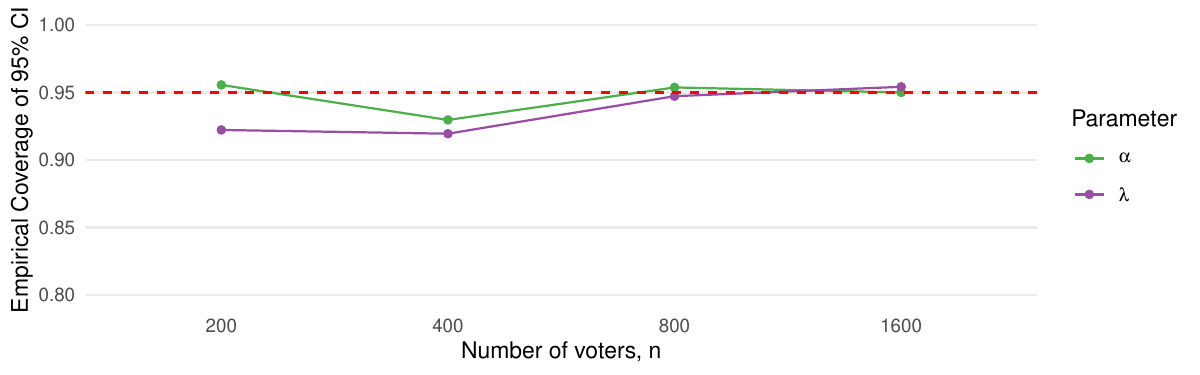}
    \caption{Empirical coverage of 95\% credible intervals (CI) for $\lambda$ and $\alpha$ across $n$.}
    \label{fig:coverage}
\end{figure}

\subsection{Sushi Preferences}\label{section:sushi}

We next fit a Bayesian Benter model using the proposed estimation algorithm to the classic sushi preferences dataset \citep{kamishima2003nantonac}. In a non-representative survey of $5{,}000$ Japanese individuals, each was asked to rank their preferences among top types of sushi: shrimp, sea eel, tuna, squid, sea urchin, salmon roe, egg, fatty tuna, tuna roll, and cucumber roll. Every individual ranked all 10 sushi types.

The dataset was fit using flat priors ($a=c=d=1$, $b=0$) and under 4 independent chains of $5{,}000$ iterations each, with the first 20\% removed as burn-in. Each independent chain took between 5-6 minutes to run. We examined trace plots for each parameter to assess mixing and convergence, which were deemed acceptable (see appendix \ref{app}). The results of the data analysis are shown in Figure \ref{fig:sushi}.
\begin{figure}[h!]
    \centering
    \includegraphics[width=.8\linewidth]{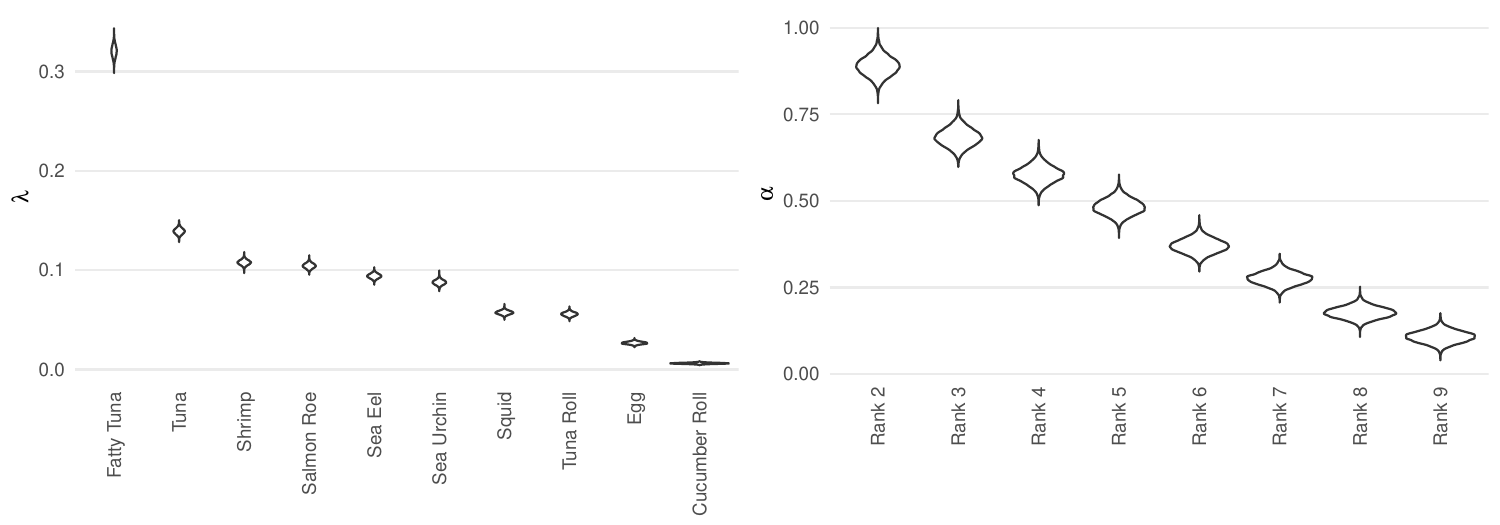}
    \caption{Estimated posterior distributions of $\lambda$ (left) and $\alpha$ (right) in a Bayesian Benter model fit to the sushi preferences dataset.}
    \label{fig:sushi}
\end{figure}

We estimate fatty tuna, tuna, and shrimp (in order) are the most-preferred sushi types. Egg and cucumber roll are the least-preferred, and notably, the only non-fish items among the choice set. We also observe consistently decreasing dampening parameters, indicating greater randomness in choice as individuals fill out their ranking.

\subsection{Irish Election Dataset}\label{section:dublin}

Last, we analyze a real 2002 voting dataset from the Dublin West constituency in Ireland. The election was carried out using the single transferable vote (STV) procedure. Out of 9 candidates, 3 were to be elected. In total, there were $29{,}988$ valid votes cast, which vary in length between 1 and 9 candidates. Only 16\% of ballots were complete. 3\% of voters ranked just one candidate, and the modal number of candidates ranked was 3 (29\%). The data is available in the \texttt{vote} R package \citep{vote_package}.

The dataset was fit using a flat prior on $\lambda$ ($a=1,\ b=0$) and a left-skewed prior on $\alpha$ ($c=10,\ d=1$) to slightly discourage $\alpha$ from approaching $0$. In the presence of nearly $30$k votes, these priors are both vague. We fit 4 independent chains of $4{,}000$ iterations each, with the first 20\% removed as burn-in. Each independent chain took about 10 minutes to run. We examined trace plots for each parameter to assess mixing and convergence, which were deemed acceptable (see appendix \ref{app}). The results of the data analysis are shown in Figure \ref{fig:dublin}.
\begin{figure}[b!]
    \centering
    \includegraphics[width=.8\linewidth]{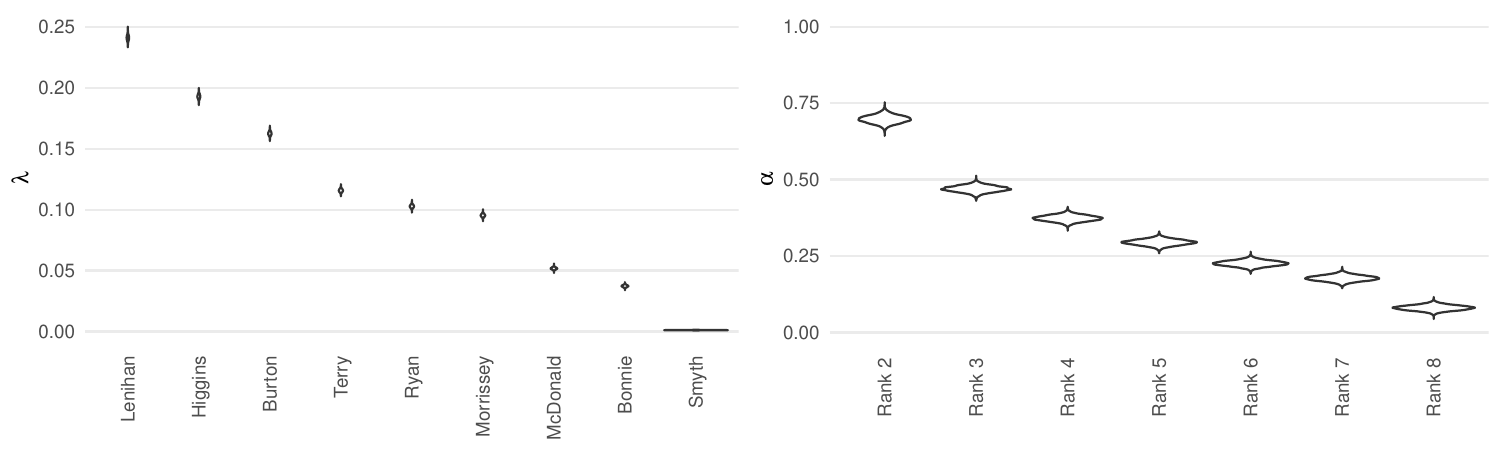}
    \caption{Estimated posterior distributions of $\lambda$ (left) and $\alpha$ (right) in a Bayesian Benter model fit to the 2002 Dublin West election dataset.}
    \label{fig:dublin}
\end{figure}

We observe candidates Lenihan, Higgins, and Burton to be in the top-3 estimated ranks. Indeed, these candidates won the election. We observe consistently decreasing dampening parameters, indicating greater randomness in choice as voters complete their ballots. The posterior mean of $\alpha_r$ is approximately $0.7$, suggesting reasonably high level of randomness when selecting even a second-choice candidate. It is important to note that these results may be the result of preference heterogeneity among voters, which is not accounted for in the present model.

\section{Discussion}\label{section:Discussion}

This paper has proposed the first efficient Bayesian estimation procedure for the Benter model for rankings. Although efficient Bayesian estimation has been proposed for the simpler Plackett-Luce model \citep{caron2012efficient, mollica2017bayesian}, previous attempts to extend such methods to the more general Benter model have been unsuccessful. The principal challenge is the presence of rank level-specific dampening parameters, which are not separable across items and yield no clear conjugate priors. Here, we present a data augmentation scheme relying on Exponential and Positive Stable (PS) distributions that permits an equivalent representation of the Benter model. Although mathematically complex, this alternative representation yields a Metropolis-Hastings-within-blocked Gibbs sampler for parameters $\lambda$ and $\alpha$. We demonstrate accuracy and computational efficiency of the proposed estimation algorithm on simulated and real datasets.

Notably, few previous works have estimated full posterior distributions when fitting data to the Benter model, instead relying on point estimates via maximum likelihood or maximum \textit{a posteriori}. We suspect this is not a symptom of the model's lack of importance or usefulness, but rather the computational effort required to fit such models. In fact, previous authors (e.g., \cite{gormley2008exploring, gormley2008mixture, vanichbuncha2017modelling,pearce2026can}) have demonstrated that the Benter is often able to fit high-dimensional ranking data better than the simpler Plackett-Luce distribution \citep{plackett1975analysis}. We note that without an efficient Gibbs sampler, Bayesian estimation requires an independent Metropolis-Hastings (MH) steps for each of the $2J-3$ free parameters per MCMC iterations. In each MH step, the Benter likelihood is computed, which is slow. To run thousands of iterations across multiple chains would be burdensome,  or even intractable, especially when the number of observations or items is large. Frequentist estimation requires the bootstrap or EM-based asymptotic approximations, which is likely to be computationally intractable for even modest-sized datasets. The present paper makes full posterior estimation of Benter models practical.

We conclude by briefly stating three future directions. First, the algorithm may yield a simple and fast method for obtaining MLE or MAP estimates of model parameters, akin to the work of \cite{mollica2017bayesian} for the Plackett-Luce model based on another Gibbs-type algorithm. Second, the algorithm could be extended to the mixture setting, such as a latent-class mixture. In particular, previous works in election modeling have found that mixtures of Benter models have fit the data particularly well \citep{gormley2008exploring,pearce2026can}. Such an extension would permit not only point estimation, but full uncertainty quantification of model parameters. Third, the Benter model can be extended to other common rank-data settings, such as incomplete rankings (e.g., pairwise comparisons), settings with ranking covariates (as in the Plackett-Luce; see, e.g.,  \cite{Gormley2006,chapaaan1982exploiting,tutz2015extended}), or rank-clustering settings \citep{pearce2025bayesian,santi2026bayesian}. Applying the efficient estimation algorithm proposed herein to these Benter extensions may be useful for practitioners.

\FloatBarrier

\section*{Acknowledgments}

The author acknowledges the use of Claude (Anthropic; Opus 4.7 and Fable 5) to assist with methodology development, improving language clarity in specific sections, and increasing code efficiency. I have reviewed and edited all AI-generated output and assume full responsibility for the content of this manuscript.

\bibliography{main}
\clearpage
\appendix

\section{Trace plots from Sections \ref{section:sushi}--\ref{section:dublin}}\label{app}

\begin{figure}[h!]
    \centering
    \includegraphics[width=\linewidth]{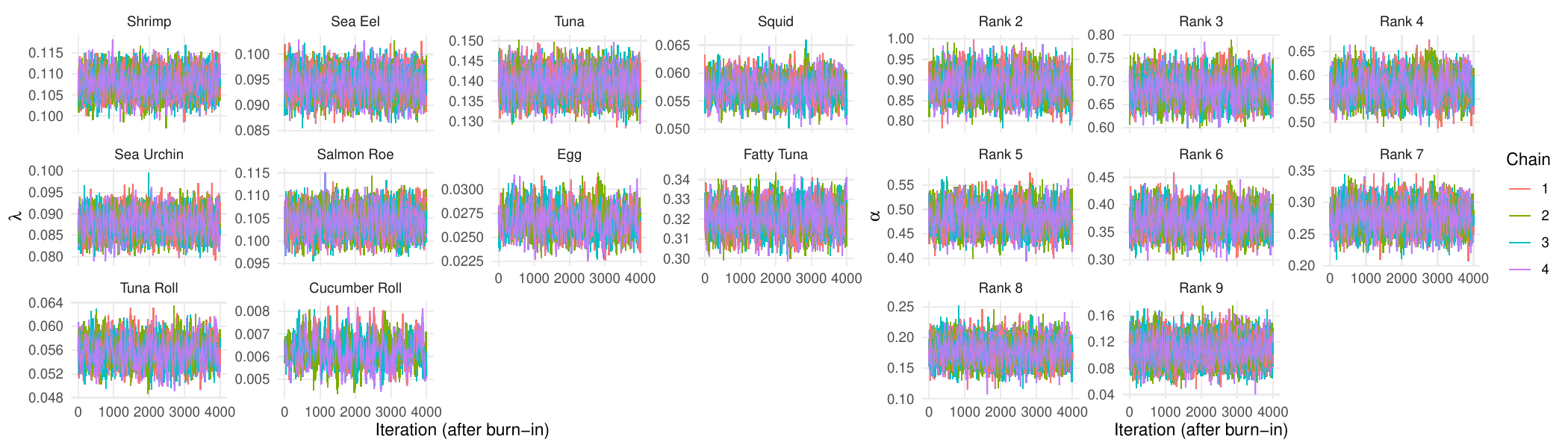}
    \caption{Trace plots of $\lambda$ (left) and $\alpha$ (right) for the sushi preferences dataset analyzed in Section \ref{section:sushi}.}
    \label{fig:sushi_trace}
\end{figure}

\begin{figure}[h!]
    \centering
    \includegraphics[width=\linewidth]{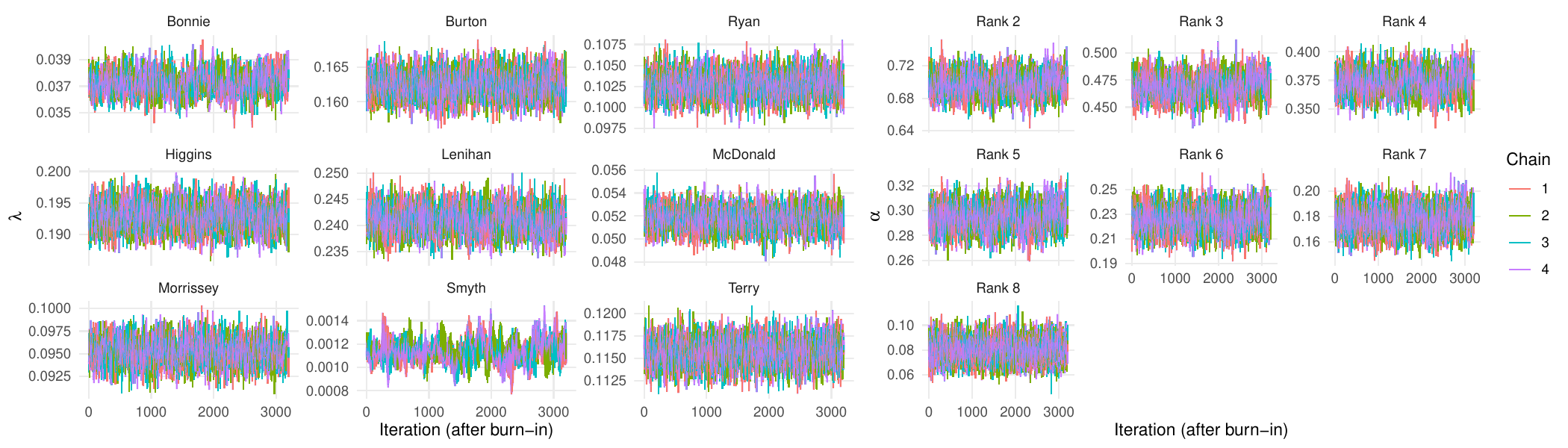}
    \caption{Trace plots of $\lambda$ (left) and $\alpha$ (right) for the 2002 Dublin West election dataset analyzed in Section \ref{section:dublin}.}
    \label{fig:dublin_trace}
\end{figure}

\end{document}